\documentclass[manuscript]{acmart}

\usepackage[linesnumbered,ruled,vlined]{algorithm2e}
\usepackage{amstext, amsmath,amsthm,color}
\usepackage{graphicx,newfloat, subcaption}
\usepackage{listings}
\usepackage{adjustbox}

\usepackage{array,enumitem}
\usepackage{makecell,multirow, adjustbox}
\usepackage{comment}

\newcolumntype{L}{>$l<$}

\newtheorem{defn}{Definition}

\newtheorem{thrm}{Theorem}
\newtheorem{prop}{Proposition}

\SetKwInput{KwInput}{Input}                % Set the Input
\SetKwInput{KwOutput}{Output}              % set the Output

\AtBeginDocument{%
  \providecommand\BibTeX{{%
    \normalfont B\kern-0.5em{\scshape i\kern-0.25em b}\kern-0.8em\TeX}}}

\setcopyright{none}
\acmConference[FAccT'2022]{}
\acmBooktitle{Seoul, South Korea}
\begin{document}

% \title{On Eliciting Non-Comparative Feedback from Human Auditors on Black-Box Classifiers: Comparative Fairness Guarantees and Learning Auditor Rules}

\title{On Learning and Enforcing Latent Assessment Models using Binary Feedback from Human Auditors Regarding Black-Box Classifiers}
\author{Mukund Telukunta}
\email{mt3qb@mst.edu}
\author{Venkata Sriram Siddhardh Nadendla}
\affiliation{%
  \institution{Missouri University of Science and Technology}
  \city{Rolla}
  \state{Missouri}
  \country{USA}
  \postcode{65401}
}

\renewcommand{\shortauthors}{M. Telukunta and V.S.S. Nadendla}
\renewcommand{\shorttitle}{Learning and Enforcing Latent Assessment Models for Human Auditor Feedback}

\begin{abstract}
Algorithmic fairness literature presents numerous mathematical notions and metrics, and also points to a tradeoff between them while satisficing some/all of them simultaneously. Furthermore, the contextual nature of fairness notions makes it difficult to automate bias evaluation in diverse algorithmic systems. Therefore, in this paper, we propose a novel model called latent assessment model (LAM) to characterize binary feedback provided by human auditors, by assuming that the auditor compares the classifier’s output to his/her own intrinsic judgment for each input. We prove that individual and/or group fairness notions are guaranteed as long as the auditor’s intrinsic judgments inherently satisfy the fairness notion at hand, and are relatively similar to the classifier's evaluations. We also demonstrate this relationship between LAM and traditional fairness notions on three well-known datasets, namely COMPAS, German credit and Adult Census Income datasets. Furthermore, we also derive the minimum number of feedback samples needed to obtain PAC learning guarantees to estimate LAM for black-box classifiers. These guarantees are also validated via training standard machine learning algorithms on real binary feedback elicited from 400 human auditors regarding COMPAS.
\end{abstract}



\keywords{Fairness, Machine Learning, Auditing, Latent Assessment Model}

\maketitle

% \newcommand\scalemath[2]{\scalebox{#1}{\mbox{\ensuremath{\displaystyle #2}}}}

% \lstset{%
% 	basicstyle={\footnotesize\ttfamily},% footnotesize acceptable for monospace
% 	numbers=left,numberstyle=\footnotesize,xleftmargin=2em,% show line numbers, remove this entire line if you don't want the numbers.
% 	aboveskip=0pt,belowskip=0pt,%
% 	showstringspaces=false,tabsize=2,breaklines=true}
% \floatstyle{ruled}
% \newfloat{listing}{tb}{lst}{}
% \floatname{listing}{Listing}

\section{Introduction}\label{sec: Introduction}
In the last five years, machine learning (ML) algorithms (e.g. classifiers) have been reported as being discriminatory with respect to the sensitive attributes (e.g. race, gender) in variety of application domains such as recommender systems in criminal justice \cite{angwin2016}, e-commerce services in online markets \cite{Harvard2016}, and life insurance premiums \cite{waxmen2018}. During this period, fundamental tradeoffs between various fairness notions have been identified by two different research groups independently \cite{chouldechova2017fair, kleinberg2016inherent}, who showed that it is impossible to satiate all fairness notions at the same time. One of the best examples is the study by ProPublica \cite{angwin2016} which demonstrated that the COMPAS risk assessment tool used in Florida's courts violates \emph{statistical parity}, a group fairness notion. However, Northpointe, the organization that designed COMPAS, presented a counterargument that the tool was indeed developed to guarantee a different fairness notion called \emph{calibration} via maintaining similar positive predictive values across different social groups. While the problem of finding new fairness notions has been active over the past decade \cite{dwork2012fairness, MoritzOpportunities, kleinberg2016inherent, zafar2017}, Mehrabi \emph{et al.} \cite{mehrabi2021survey} highlights that there are too many notions for measuring unfairness with few differences to describe them. Moreover, with the existence of trade-offs between fairness notions, selecting a notion is highly situation/context dependent \cite{binns2018fairness}. This decision of choosing the metric comes with huge responsibility since it restraints one point of view while silencing another. Rovatsos \emph{et al.} \cite{rovatsos2019landscape} says that the burden of selecting the appropriate metric is presently in the hands of practitioners, the majority of whom are unfamiliar with fairness research, which is a high bar given that society as a whole has yet to decide which ethical principles to prioritise. 
% It is still unclear which particular fairness notion is most appropriate for the given context and application \cite{Gajane2017}. 

This inability to automatically identify an appropriate fairness notion necessitates the involvement of human auditors within the fair-ML pipeline so as to mitigate any social discrimination due to ML-based systems. More specifically, we need a crowd-auditing platform where human auditors can review the outcomes generated by an ML algorithm, and evaluate its biases using an apt fairness notions depending on the context and application. However, there are many practical challenges in designing an effective crowd-auditing platforms, some of which are given below: (i) \emph{physical limits of human auditors} in terms of their ability to give feedback to large datasets, (ii) \emph{data labeling is expensive} especially when the number of possible inputs (e.g. types of people affected by the system) is quite large, (iii) \emph{human feedback modeling}, where platform learns the fairness notion from limited feedback collected from human auditors, (iv) \emph{opinion heterogeneity and aggregation}, where different human auditors may have non-aligned opinion regarding the appropriate fairness notion, and above all, (v) \emph{biased auditors} whose judgements are poisoned by their inherent biases.

% faces many challenges such as identifying an appopriate fairness notion via formal characterization of human judgements, especially when they are often susceptible to various types of biases. Furthermore, data collection is expensive, which allows to collect limited feedback from the auditors.

This paper focuses primarily on human feedback modeling from preliminary feedback, so that the crowd-audit system can reproduce auditor's judgements artificially over remaining input possibilities in a given ML-based system. In the past, several researchers have attempted to model human perception of fairness, but have always tried to fit their revealed feedback to one of the existing fairness notions. For instance, Jung \emph{et al.} in \cite{jung2019eliciting} estimated the task-based similarity metric used in individual fairness notion \cite{dwork2012fairness} based on feedback elicited from auditors regarding how a given pair of individuals have been treated. Their algorithm minimizes the classification subject to the violations of specified pair. Similarly, Gillen \emph{et al.} \cite{gillen2018online} assumes the existence of an auditor who is capable of identifying unfairness given pair of individuals when the underlying similarity metric is Mahalanobis distance. From the perspective of group fairness notions, Srivastava \emph{et al.} in \cite{srivastava2019mathematical} performed an experiment by asking participants to choose among two different models to identify which notion of fairness (demographic parity or equalized odds) best captures people's perception in the context of both risk assessment and medical applications. Likewise, Harrison \emph{et al.} \cite{harrison2020empirical} surveyed 502 workers on Amazon's Mturk platform and observed a preference towards \emph{equal opportunity}. Note that all these experiments \cite{srivastava2019mathematical, harrison2020empirical} ask participants to reveal their analysis with respect to a specific fairness notion in the context of given sensitive attributes (e.g. race, gender). In fact, these papers \cite{srivastava2019mathematical, harrison2020empirical} themselves clearly point their limitation regarding generalization of their results. On the contrary, in this paper, we impose no such restrictions on the auditor in constructing their feedback with respect to satiating any specific fairness notion. However, we assume that the expert auditor employs an intrinsic fair relation (which is unknown) to evaluate a given data tuple. 

At the same time, several researchers have also attempted to identify fairness notions that are preferred by human auditors under different application contexts. Saxena \emph{et al.} in \cite{saxena2019fairness} investigated how people perceive the fitness of three different algorithmic-fairness notions in the context of loan decisions, and whether fairness perceptions change with the addition of sensitive information. Another interesting work is by Grgi\'c-Hla\v{c}a \emph{et al.} in \cite{grgic2018human}, who discovered that people's fairness concerns are typically multi-dimensional (relevance, reliability, and volitionality), especially when binary feedback was elicited. This means that modeling human feedback should consider several factors beyond social discrimination. Dressel and Farid in \cite{dressel2018accuracy} showed that COMPAS is as accurate and fair as that of untrained human auditors regarding predicting recidivism scores. Awad \emph{et al.} in \cite{awad2018moral} developed a novel experimental platform to explore the moral dilemmas faced by autonomous vehicles, agnostic to the traditional fairness notions considered in the algorithmic fairness literature. Specifically, the platform gathered 40 million decisions in ten languages from millions of people in 233 countries and territories, and categorized demographic and cross-cultural ethical variations in moral preferences globally into three major morality clusters which need not necessarily align with each other. Similar effort was also made by Lee \emph{et al.} in the context of food donation distribution as well \cite{lee2019webuildai}. On the other hand, Yaghini \emph{et al.} \cite{yaghini2021human} proposed a novel representation of family of fairness notions, equality of opportunity (EOP), which requires that for individuals with similar desert, the distribution of utility should be the same. In addition, they learned the proposed notion of fairness by eliciting human judgments and utilizing the data to optimize the EOP parameters for criminal risk assessment context. Results show that the proposed notion performs better than existing notions of algorithmic fairness in terms of equalizing utility distribution across groups. A major drawback of these approaches is that the demographics of the participants involved in the experiments \cite{yaghini2021human, grgic2018human, harrison2020empirical, saxena2019fairness} are not evenly distributed. For instance, the conducted experiments ask how models treated Caucasians and African-Americans, but there were insufficient non-Caucasian participants to assess whether there was a relationship between the participant's own demographics and what group was disadvantaged. Moreover, the participants are presented with multiple questions in the existing literature which cannot be scaled for larger decision-based models \cite{yaghini2021human}. 

In contrast to the previous works discussed above, this paper proposes a novel \emph{latent assessment model} (LAM) when human auditors reveal binary feedback (fair or not) for the given data tuple collected from the ML algorithm. Unlike most of the past literature on human perception of fairness, we assume that human auditors are given the freedom to reveal binary feedback (fair, or unfair), while not being forced to follow any specific fairness notion artificially. Although our binary feedback structure is similar to that discussed in \cite{gillen2018online}, we do not assume that this feedback is necessarily aligned with any one fairness notion. 
% In this paper, we show that both individual fairness and group fairness notions are based only on \emph{comparative justice} principles, i.e. they compare individuals/groups based on the similarity of their treatment/outcomes. 
Inspired by \emph{non-comparative justice} principles \cite{levine2005comparative, feinberg1974noncomparative, montague1980comparative}, where every individual is to be treated precisely based on their own personal attributes and merits regardless of how other individuals are treated/affected by the same service, we demonstrate that LAM is well suited to characterize people's judgements in the real world. 
% In other words, we assume that each auditor assesses the input attributes, constructs a desired outcome, and compares the algorithm's outcome with their own assessment. For example, in order to evaluate the algorithm behind college admissions, an auditor might consider grades/marks, projects, and internships in an individual's profile to construct their desired admission decision and compares it with the algorithm's outcome. 
% Specifically, the proposed LAM model  If both the system's and auditor's outcomes are (dis)similar, then the system will be deemed (un)fair by the auditor.
% We show that the proposed LAM is fundamental to achieving both individual and group fairness notions. 
We prove that a system/entity satisfies individual and/or group fairness notions if the auditor exhibits LAM in his/her fairness evaluations. We also show that converse holds true in the case of individual fairness. We also validate these relationships with traditional fairness notions on three real datasets, namely COMPAS, Adult Income Census and German credit datasets. Since both the system and auditor rules are hidden, we compute PAC learning guarantees on algorithmic auditing based on binary feedback obtained from human auditors. We compare various learning frameworks (e.g. logistic regression, support vector machines and decision trees) to estimate auditor's intrinsic judgements and their feedback on a real human feedback dataset which was collected by Dressel and Farid in \cite{dressel2018accuracy}.

\section{Preliminaries: Traditional Fairness Notions \label{sec: fairness notions}}
Much of the existing work on algorithmic fairness has focused on two important notions: \emph{group fairness} and \emph{individual fairness}. The notion of group fairness seek for parity of some statistical measure across all the protected attributes present in the data. Different versions of group-conditional metrics led to different group definitions of fairness. For example, \emph{statistical parity} \cite{dwork2012fairness} can be formally defined as follows.
\begin{defn}[Coarse Statistical Parity]
Given protected attributes $A$, if $f$ is a predictor and $Y = f(X)$, where, $X$ is the multi-attribute variable and $Y$ is the outcome, $f$ achieves statistical parity when
\begin{equation}
\mathbb{P}[f(x) = 1 \ | \ A = a] - \mathbb{P}[f(x) = 1 \ | \ A = a'] \leq \delta
\end{equation}
holds true for all $y \in Y$, $x \in X$, and $a, a' \in A$.
\label{Defn: GF}
\end{defn}
In other words, statistical parity seeks to achieve equal probability of individuals who are predicted to be positive across different groups. Similarly, the notion of \emph{equal opportunity} \cite{MoritzOpportunities} states that the true positive rate should be the same for all the groups. A coarser version of equal opportunity can be defined as follows.
\begin{defn}[Coarse Equal Opportunity]
We say that a binary predictor $f$ satisfies equal opportunity with respect to set of protected attributes $A$ and outcome, $Y = f(x)$, if
\begin{equation}
\mathbb{P}[f(x) = 1 \ | \ y = 1, A = a] - \mathbb{P}[f(x) = 1 \ | \ y = 1, A = a'] \leq \delta.
\end{equation}
holds true for all $a, a' \in A$.
\label{Defn: EO}
\end{defn}

Another example for group-conditional fairness notion is \emph{calibration} \cite{kleinberg2016inherent, chouldechova2017fair} which ensures equal positive predictive value across different groups.
\begin{defn}[Coarse Calibration]
A binary predictor $f$ satisfies calibration given a set of protected attributes $A$ and outcome, $Y = f(x)$, if
\begin{equation}
\mathbb{P}[y = 1 \ | \ f(x) = 1, A = a] - \mathbb{P}[y = 1 \ | \ f(x) = 1, A = a'] \leq \delta.
\end{equation}
holds true for all $a, a' \in A$.
\label{Defn: Calibration}
\end{defn}
On the contrary, the notion of individual fairness \cite{dwork2012fairness} states that similar individuals should be treated similarly concerning a particular task. The notion can be defined as follows.
\begin{defn}[Individual Fairness]
Given any two individuals $x_i, x_j \in \mathcal{X}$ with $\mathcal{D} \left( x_i,  x_j \right) \leq \kappa$, then the classifier $f$ is $(\kappa, \delta)$-individually fair if $d \Big( f(x_i), f(x_j) \Big) \leq \delta$. 
\label{Defn: IF}
\end{defn}
Researchers have proposed several other fairness notions in the past decade. For more details, interested readers may refer to a detailed survey on fairness in algorithmic decision-making in \cite{caton2020fairness, chouldechova2018frontiers, mehrabi2021survey, pessach2020algorithmic}.

\section{Latent Assessment Model: Relation with Traditional Fairness Notions\label{sec: LAM}}
Consider an expert auditor who, when presented with a data tuple $(x, y)$ where $x \in \mathcal{X}$ is the input given to ML model $g$ and $\tilde{y} = g(x) \in \mathcal{Y}$ is the output label, reveals binary feedback $s$ as shown below.  
\begin{defn}[$\epsilon$-Latent Assessment Model]
Let $f$ denote the fair assessment adopted by the expert auditor where, $y = f(x)$ is the subjective evaluation for the input $x$. If $g$ is an unknown system, i.e. the system's label is defined as $\tilde{y} = g(x)$, then $\epsilon$-latent assessment model ($\epsilon$-LAM) of the expert auditor is a tuple $(\mathcal{X}, \mathcal{Y}, d, f, \epsilon)$ such that
% Let $f$ denote a fair assessment (i.e. input-output relationship) of a given system, i.e. $y = f(x)$, which is evaluated subjectively by an expert auditor. If $g$ is an alternative representation, i.e. $\tilde{y} = g(x)$ (e.g. classification algorithms minimizing loss function, recommender systems), then $g$ is called $\epsilon-$noncomparatively fair w.r.t. $f$ if 
\begin{equation}
\displaystyle d \left( g(x), f(x) \right) < \epsilon, \text{ for all } x \in \mathcal{X}
\end{equation}
and the binary judgment presented by the auditor is given as follows.
\begin{equation}
s = \begin{cases}
1, & \text{if } \ d(g(x), f(x)) \geq \epsilon,
\\[1ex]
0, & \text{otherwise.}
\end{cases}
\label{Eqn: Auditor Judgement}
\end{equation}
\label{Defn: Noncomparative Fairness}
\end{defn}

In other words, we believe that the auditor compares their intrinsic classification, $y = f(x)$, with the ML model's output label and presents a binary judgment $s$. For example, consider an individual who committed felony and has multiple prior offences, received low recidivism score from a risk assessment tool. The expert auditor evaluates the individual intrinsically and believes that he/she should receive a higher recidivism score. Thus, the auditor would judge the tool's output as unfair. In this paper, the fair relation $f$ employed by the expert auditor is unknown. Therefore, our goal is to learn the proposed $\epsilon$-LAM using statistical learning techniques. Although $\epsilon$-LAM comprises of three unknowns in practice, namely $d$, $f$ and $\epsilon$, we assume that the distance metric $d$ used by the auditor is known.   
\begin{figure}[!t]
\centering
\includegraphics[width=0.7\textwidth]{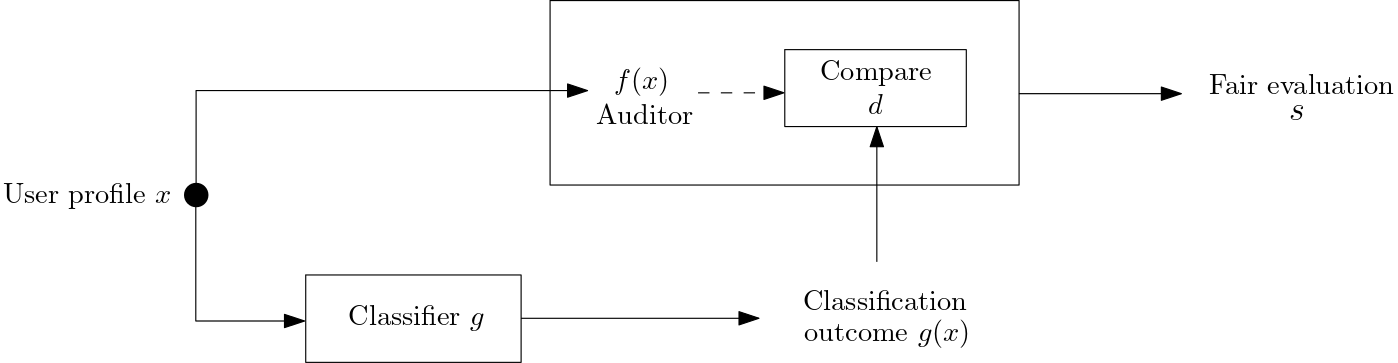}
\caption{Latent Assessment Model of the Expert Auditor}
\label{fig:NCF Architecture}
\vspace{-2ex}
\end{figure}

Before we estimate $f$, note that auditors exhibit different types of biases in the real-world. Examples include confirmation bias, hindsight bias, anchoring bias, racial and gender bias. Our proposed model, LAM, captures auditor biases through the function $f$ in Equation \eqref{Eqn: Auditor Judgement}. Therefore, in the remainder of this section, we investigate the relationship between our proposed LAM and traditional fairness notions. By doing so, we can estimate any given auditor's performance in terms of multiple fairness notions.

\subsection{Relation with Individual Fairness}\label{Sec - Individual Fairness}
In the following proposition, we show how the relation $g$ can be evaluated based on the notion of $(\kappa, \delta)$-individual fairness, when $g$ is $\epsilon$-LAM with respect to auditor's assessment $f$.
\begin{prop}
$g$ is $(\kappa, 2 \epsilon + \delta)$-individually fair, if $g$ is     $\epsilon$-LAM with respect to $f$, and $f$ is $(\kappa, \delta)$-individually fair.
\label{Prop: (e, k, d)-IF}
\end{prop}
\begin{proof}
Given $(x_1, y_1)$ and $(x_2, y_2) \in \mathcal{X} \times \mathcal{Y}$ such that $\mathcal{D} \left( x_1,  x_2 \right) \leq \kappa$ (the two individuals are $\kappa$-similar), then $f$ is $(\kappa, \delta)$-individually fair if $d \Big( f(x_1), f(x_2) \Big) < \delta$. However, note that if $g$ is $\epsilon$-LAM with respect to $f$, then $ d \Big( g(x_1), f(x_1) \Big) < \epsilon$ and $d \Big( g(x_2), f(x_2) \Big) < \epsilon$. Therefore, by applying a chain of triangle inequalities, we obtain
\begin{equation}
\begin{split}
& d \Big( g(x_1), g(x_2) \Big) \leq d \Big( g(x_1), f(x_1) \Big) + d \Big(  f(x_1), f(x_2) \Big) + d \Big( f(x_2), g(x_2) \Big) \ < 2 \epsilon + \delta.
\end{split}
\label{Eqn: Prop1 - Final eqn}
\end{equation}
\end{proof}

We illustrate this result using the following example from the banking domain. Consider two individuals who are looking to apply for a loan. A banking system would evaluate both the applications via collecting information such as gender, race, address, credit history, collateral, and his/her ability to pay back. At the same time, consider an auditor who makes fairness judgements based on the rule: "If he/she has cleared all the debts and possesses reasonably valued collateral, the loan must be granted". Given that the auditor treats any two similar individuals similarly, the auditor satisfies individual fairness. Hence, if the evaluation of the banking system is relatively similar to the auditor's fair relation, from Proposition \ref{Prop: (e, k, d)-IF}, the banking system is also individually fair.

\begin{prop}
If $f$ is not $(\kappa, \delta)$-individually fair and if $g$ is $\epsilon$-LAM with respect to $f$, then $g$ is not $(\kappa, \delta - 2\epsilon)$-individually fair.
\label{Prop: IF - converse}
\end{prop}
\begin{proof}
If $f$ is not individually fair, then for some input pair $(x_1, x_2)$ such that $\mathcal{D}(x_1, x_2) < \kappa$, we have $d\left(f(x_1), f(x_2)\right) > \delta$ for all $\kappa, \delta \in \mathbb{R}$. However, note that if $g$ is $\epsilon$-LAM with respect to $f$, then $d\left(g(x_1), f(x_1)\right) < \epsilon$ and $d\left(g(x_2), f(x_2)\right) < \epsilon$. Therefore, by applying a chain of triangle inequalities, we have
\begin{equation}
d\left(f(x_1), f(x_2)\right) \leq d\left(g(x_1), f(x_1)\right) + d\left(g(x_2), f(x_2)\right) + d\left(g(x_1), g(x_2)\right)
\end{equation}
Substituting the bounds of $d\left(g(x_2), f(x_2)\right)$ and  $d\left(g(x_1), g(x_2)\right)$ we get
\begin{equation}
\begin{split}
2 \epsilon + d\left(g(x_1), g(x_2)\right) & > d\left(g(x_1), f(x_1)\right) + d\left(g(x_2), f(x_2)\right) + d\left(g(x_1), g(x_2)\right)
\\[1ex]
& \geq d\left(f(x_1), f(x_2)\right) > \delta
\end{split}
\end{equation}
for all $\delta \in \mathbb{R}$. Therefore, we also have
\begin{equation}
d\left(g(x_1), g(x_2)\right) > \delta - 2\epsilon.
\end{equation}
\end{proof}

Consider the earlier example of banking where, there are two individuals, $A$ and $B$, who possess the same degree of merit. Imagine that the bank approves $A$'s loan application and denies $B$. This outcome remains the same as per the auditor's fair relation. Imagine further that neither $A$ nor $B$ merits the outcome. Though both banking's evaluation and auditor's judgements seem to be similar, they violate the precept, "treat similar individuals similarly". Moreover, $A$ is treated in a way that $A$ does not merit. Hence, we can assert that banking evaluation does not satisfy individual fairness.      

\subsection{Relation with Group Fairness}\label{sec: Group Fairness}
As discussed earlier, group fairness notions compare certain probabilistic measure across two protected groups. In the remaining section, we will focus on the relationship between group fairness notions and our proposed LAM. For the sake of convenience, let us denote $p_{x,y}(g,a) = \mathbb{P}[g(x) = y \ | \ A = a]$.
\begin{prop}
Given that the probability distributions are $M$-Lipschitz continuous over all possible $f$ and $g$ functions, $g$ satisfies $(2M\epsilon + \delta)$-statistical parity, if $g$ is $\epsilon$-LAM with respect to $f$, and $f$ satisfies $\delta$-statistical parity.
\label{Prop: GF}
\end{prop}
\begin{proof}
Given the set of protected attributes $\mathcal{A}$, since $f$ satisfies $\delta$-statistical parity, we have $ || p_{x,y}(f,a) - p_{x,y}(f,a') || < \delta$ for all $a, a' \in \mathcal{A}$. Then, we have
\begin{equation}
\begin{split}
& p_{x,y}(g, a) - p_{x,y}(g, a') = \left[ p_{x,y}(g, a) - p_{x,y}(f, a) \right] + \left[ p_{x,y}(g, a') - p_{x,y}(f, a') \right] + \left[ p_{x,y}(f, a) - p_{x,y}(f, a') \right]
\end{split}
\end{equation}
Assuming $M$-Lipschitz continuity over all $f(x)$, $g(x)$, we have $|| p_{x,y}(g, a) - p_{x,y}(f, a) || < M \cdot \epsilon$, since $d (g(x), f(x) ) < \epsilon$. Combining all the inequalities, we have 
\begin{equation}
|| p_{x,y}(g, a) - p_{x,y}(g, a') || < 2 M \epsilon + \delta.
\label{Res: GF}
\end{equation}
\end{proof}
Again, consider the earlier example of loan approvals to illustrate the above proposition. Consider that there exists two groups which are classified based income - low and high. The banking system builds a credit model based purely. Moreover, the system may decide to use different requirement levels - low interest or default to low income group, so that the percentage of people getting a loan in low-income group is equal to the percentage of people getting a loan in high-income group. Now, suppose an auditor presents fair judgements based on the rule: “If Group A has a FICO credit score of 550 and cleared all the debts, the loan must be granted. If Group B has a FICO score of 700 and has valuable collateral, grant the loan”. Note that, the auditor's fair relation is somewhat similar to that of the bank's policy. Since the bank's policy is known to be statistically fair, the auditor is also unbiased from a group fairness perspective.  

Similarly, the following two propositions identify the relationship between our proposed non-comparative notion and two other group fairness notions, namely coarse equal opportunity and coarse calibration.
\begin{prop}
Given that the probability distributions are $M$-Lipschitz continuous over all possible $f$ and $g$ functions, $g$ satisfies $(2M\epsilon + \delta)$-equal opportunity, if $g$ is $\epsilon$-LAM with respect to $f$, and $f$ satisfies $\delta$-equal opportunity.
\label{Prop: EO}
\end{prop}
\begin{proof}
The proof is similar to that of Proposition \ref{Prop: GF}. Therefore, for the sake of brevity, the proof is not included. 
\end{proof}

\begin{prop}
Given that the probability distributions are $M$-Lipschitz continuous over all possible $f$ and $g$ functions, $g$ satisfies $(2M\epsilon + \delta)$-calibration, if $g$ is $\epsilon$-LAM with respect to $f$, and $f$ satisfies $\delta$-calibration.
\label{Prop: Calibration}
\end{prop}
\begin{proof}
The proof is similar to that of Proposition \ref{Prop: GF}. Therefore, for the sake of brevity, the proof is not included. 
\end{proof}

\section{Learning LAM from Elicited Feedback}
In practice, the auditor's intrinsic rule $f$ is not revealed in his/her feedback. Therefore, we need to compute the intrinsic rule $\hat{f}$ in order to reproduce auditor's judgements for other input possibilities. At the same time, the classifier is typically available to the bias-evaluation platform as a black-box system, i.e., $g$ is also unknown to the bias-evaluation platform. In other words, a practical bias-evaluation platform is expected to compute $\hat{f}$ and identify an appropriate fairness notion for the given context so that the bias evaluation platform can algorithmically evaluate bias in a system with a large input space. 

\subsection{PAC Learning Guarantees \label{sec: PAC Guarantees}}
\begin{defn}[PAC Learnability]
We say that the classifier $f$ is PAC-learnable if there exists $N > 0$, $\epsilon > 0$, $\delta > 0$, and an algorithm $\mathcal{A}$ which receives $n \geq N$ i.i.d. samples from distribution $D$ as input, and outputs an estimated classifier $\hat{f}_n$ with at least $1-\delta$ probability such that $d(f, \hat{f}_n) \leq \epsilon$. 
\end{defn}

In the following theorem, we provide guarantees for the algorithmic non-comparative fairness evaluations, based on estimated rules $\hat{f}$ and $\hat{g}$. 

\begin{thrm}
Let $N$ denote the minimum number of samples needed to guarantee non-comparative fairness empirically, i.e. $\mathbb{P} \left( d(\hat{g}_n, \hat{f}_n) < \epsilon_{ncf} \right) > 1 - \delta_{ncf}$ for some $\epsilon_{ncf} > 0$ and $\delta_{ncf} > 0$. Then, for any auditor's intrinsic rule $f$ and classifier $g$, there exists some $0< N_f, N_g < N$, $\epsilon_g, \epsilon_f, \epsilon, \delta, \delta_g, \delta_f > 0$ such that $\epsilon_g + \epsilon_f + \epsilon_{ncf} < \epsilon$ and $\delta_g + \delta_f + \delta_{ncf} < 2 + \delta$, and an algorithm $\mathcal{A}$ that receives i.i.d. samples $\{ (x_i, y_i, z_i) \}_{i = 1}^n$ as input, and outputs rules $\hat{f}_n$ and $\hat{g}_n$ with a probability of $d(f(x), g(x)) < \epsilon$ being at least $1- \delta$, only when
\begin{equation}
n \geq N \triangleq \min_{\epsilon_1, \epsilon_2, \delta_1, \delta_2} (\max \{N_g, N_f\}),
\end{equation}
where $N_f$ and $N_g$ satisfy PAC learning bounds for $f$ and $g$. 
\end{thrm}
\begin{proof}
Our goal is to ensure that
\begin{equation}
\mathbb{P}(d(f(x), g(x)) < \epsilon) \geq 1-\lambda
\label{Eqn: PAC-NCF}
\end{equation}
for some small $\epsilon > 0$ and $\delta > 0$. Assuming that there exists some $0 < \epsilon_g < \epsilon$ and $0 < \epsilon_f < \epsilon$, we obtain an upper bound to LHS in the above equation using triangle inequality, as shown below:
% For $d(f(x), g(x)) = ||g(x) - f(x)||$, 
% \begin{equation}
% \scalemath{0.85}{
% \begin{split}
% & \quad \mathbb{P} \Big( d(g(x), f(x)) < \epsilon \Big)
% \\[1ex]
% & = \mathbb{P} \Big( ||g(x) - \hat{g}_n(x) - f(x) + \hat{f}_n(x) - \hat{f}_n(x) + \hat{g}_n(x)|| < \epsilon \Big)
% \end{split}
% }
% \end{equation}
% where, $\hat{g} = \mathcal{G}(\mathcal{X}_{N_g})$ and $\hat{f} = \mathcal{F}(\mathcal{X}_{N_f})$ indicate the estimates of $g$ and $f$ for the given learning algorithms $\mathcal{G}$ and $\mathcal{F}$ respectively.
\begin{equation}
\begin{split}
& \mathbb{P} \Big( d(g(x), f(x)) < \epsilon \Big) \leq \mathbb{P} \Big( d(g(x),\hat{g}_n(x)) < \epsilon_g \Big) + \mathbb{P} \Big(d(f(x), \hat{f}_n(x)) < \epsilon_f \Big) + \mathbb{P} \Big( d(\hat{f}_n(x), \hat{g}_n(x)) < \tau \Big)
\end{split}
\label{Eqn: PAC-NCF-2}
\end{equation}
where $\tau = \epsilon-\epsilon_g-\epsilon_f$.

Note that the first and the second probability terms correspond to PAC learnability guarantees of $g$ and $f$ respectively. Let $N_g(\epsilon_g, \delta_g)$ and $N_f(\epsilon_f, \delta_f)$ denote the minimum samples needed to guarantee PAC learnability at $g$ and $f$ respectively. In other words, the maximum of the two numbers will guarantee PAC learnability of both $g$ and $f$, i.e. $N(\epsilon_g, \epsilon_f, \delta_g, \delta_f) = \displaystyle \max \{ N_g(\epsilon_g,\delta_g), N_f(\epsilon_f, \delta_f) \}$. However, for the valid inequality in Equation \eqref{Eqn: PAC-NCF-2}, it is also essential to split $\delta$ between PAC learning bounds for $g$ and $f$, as well as the probability corresponding to empirical non-comparative fairness $\mathbb{P} \Big( d(\hat{f}_n(x), \hat{g}_n(x)) < \epsilon_{NCF} \Big) $. In other words, the split is valid only when
\begin{equation}
\begin{array}{lcl}
\epsilon_g + \epsilon_f + \epsilon_{ncf} < \epsilon, \text{ and }
\\[1ex]
(1-\delta_g) + (1-\delta_f) + (1-\delta_{ncf}) > 1 - \delta.
\end{array}
\end{equation}
Then, $N$ can be minimized by choosing an appropriate split $\epsilon_g$, $\epsilon_f$, $\delta_g$ and $\delta_f$ to obtain the necessary guarantee stated in Equation \eqref{Eqn: PAC-NCF}.
\end{proof}

Though we computed the minimum number of samples required, there should exist an algorithm $\mathcal{A}$ that receives i.i.d samples as input to estimate the intrinsic fair rule $f$ of the expert auditor. Therefore, we validate our findings using different ML models to learn and the predict auditor responses.

% To estimate the minimum number of samples, we divide the training set into multiple bundles of different sizes in increasing order. The algorithm $\mathcal{A}$ trains over each bundle with the loss function $l(f, \hat{f})$ and estimates the auditor's responses for the instances present in $\mathcal{X}_{test}$. 

\section{Results and Discussion \label{sec: Validation}}
In this section, we will validate our proposed $\epsilon$-LAM using simulation as well as real human audit data. First, we simulate human assessment by arbitrarily defining fair relation on three real-world datasets. We demonstrate the performance of the simulated assessment with respect to different fairness notions. Furthermore, using feedback from real human auditors, we perform preliminary analysis and learn their intrinsic relations. Though our framework specifies that the auditor reveals a binary judgement $s \in \{0, 1\}$, for the sake of practical evaluation, we assume that he/she reveals the exact classification of the given input in the same space as the respective system. 

\subsection{Simulating Human Assessments on Real Datasets}
In this section, we validate our theoretical findings using real-world data. We experiment with three well-known datasets, each of which are pre-processed as follows:
\begin{enumerate}
\item \textbf{\emph{ProPublica's COMPAS dataset \cite{larson2016}:}} 
In this paper, we perform same preprocessing as the original analysis of ProPublica. The races in the dataset are only restricted to African-American, Caucasian, and other. We consider \emph{females} and \emph{Caucasians} as privileged groups, and 0 (least likely) as the favourable outcome. Since the feature \emph{age} is continuous, we create different age groups (e.g. 25-45 or $>45$) and rename the features as \emph{age category}. Similar grouping is also performed on \emph{priors count} as well. To encode the categorical features (age category, priors count, and charge degree), we generated dummies for each feature and converted categorical columns to columns of 0s and 1s. Note that we also consider \emph{decile score} as an output feature to test our approach on $M$-ary classifiers (binary classifier if $M = 2$, and non-binary classifier if $M > 2$).

\item \textbf{\emph{German credit data \cite{germanCredit}:}} 
In this dataset, we consider credit history, savings, employment, sex, and age as input features. Moreover, we categorize \emph{age} into two groups: young ($<$ 26) and old ($>=$26). We assume that males and older individuals as privileged groups and 1 (good credit risk) as favourable outcome.

\item \textbf{\emph{Adult income dataset \cite{adultIcome}:}} The objective is to predict whether the income of an individual is $>\$50K$ or $<\$50K$. The input features include age, sex, race, and education. In the preprocessing phase, the continuous feature \emph{age} is transformed into different groups of ages (0-10, 11-20, and so on). For the feature \emph{race}, we limited the labels to binary by mapping `White' to 1 and all other races to 0. We have 32561 data tuples in total.
\end{enumerate}

\subsubsection{Evaluating Individual Fairness}
Note that, individual fairness notions rely on a similarity metric $\mathcal{D}$ between two individuals. Since the attributes in real-world datasets are correlated to one another, we consider Mahalanobis distance to compute the similarity between two randomly picked individuals, because it measures distances between points considering how the rest of the datapoints are distributed. The squared Mahalanobis distance can be defined as follows.
\begin{equation}
\mathcal{D}^2 = (\boldsymbol{x}_i - \boldsymbol{x}_j)^T C^{-1} (\boldsymbol{x}_i - \boldsymbol{x}_j)
\end{equation}
where $\boldsymbol{x}_i, \boldsymbol{x}_j$ are observations/rows in a dataset and $C$ is positive semi-definite covariance matrix. Initially, we compute the covariance matrix which summarizes the variance of the dataset.

We evaluate the COMPAS dataset for individual fairness with \emph{decile score} (on a scale of 1 to 10) as the output feature. In COMPAS dataset, there are 106 different clusters of individuals, where every individual in a cluster has 0 Mahalanobis distance with every other individual in that particular cluster. In other words, every individual in a cluster are similar to each other. Whereas, individuals from different clusters are dissimilar. Table \ref{tab:clusters example} shows 10 such clusters in the dataset. The auditor is said to be individually fair if the output given by him/her is same for every individual in that cluster. 

\begin{table}[h]
    \centering
    \begin{tabular}{|c|c|c|c|c|}
    \hline
    Sex & Age & Race & Prior Offenses & Charge Degree \\
    \hline
    Female & 25-45 & Caucasian & 1 to 3 & M \\
    Male & 25-45 & Other & 0 & F \\
    Female & Greater than 45 & African-American & 0 & F \\  
    Male & 25 - 45 & Other & More than 3 & M \\      
    Male & Greater than 45 & Other & 1 to 3 & M \\
    Male & Greater than 45 & African-American & More than 3 & M \\ 
    Female & Less than 25 & Other & 0 & M \\ 
    Female & Less than 25 & Other & More than 3 & M \\
    Female & 25 - 45 & Caucasian & 0 & M \\ 
    Male & 25 - 45 & Caucasian & More than 3 & M \\
    \hline
    \end{tabular}
    \caption{Example of 10 different clusters present in COMPAS dataset}
    \label{tab:clusters example}
\end{table}
As mentioned earlier, the auditor presents the true classification of the given individual in the same output space as the COMPAS dataset. Consider the following illustrative example for auditor's intrinsic rule to evaluate COMPAS dataset.
\begin{center}
\begin{adjustbox}{width=0.6\textwidth}
$
\tilde{f}_1(x) = 
\begin{cases}
1, & \text{if $x$.priors-count = 0 and $x$.charge-degree = M}
\\[0.5ex]
2, & \text{if $x$.priors-count = 0 and $x$.charge-degree = F} 
\\[0.5ex]
3, & \text{if $x$.priors-count $\in [1, 3]$ and $x$.charge-degree = F}
\\[0.5ex]
4, & \text{if $x$.priors-count $\in [4, 7]$, $x$.age-category $> 45$ and $x$.charge-degree = M} 
\\[0.5ex]
5, & \text{if $x$.priors-count $\in [4, 7]$, $x$.age-category $> 45$ and $x$.charge-degree = F} 
\\[0.5ex]
6, & \text{if $x$.priors-count $\in [4, 7]$ and $x$.age-category $\in [16-45]$} 
\\[0.5ex]
7, & \text{if $x$.priors-count $\in [8, 15]$ and $x$.age-category $> 45$}
\\[0.5ex]
8, & \text{if $x$.priors-count $\in [8, 15]$ and $x$.age-category $\in [16-45]$}
\\[0.5ex]
9, & \text{if $x$.priors-count $\in [16, 24]$ and $x$.age-category $> 45$}
\\[0.5ex]
10, & \text{if $x$.priors-count $> 24$}
\end{cases}
$
\end{adjustbox}
\end{center}
Based on the above intrinsic rule, the auditor seems to assess every individual in a cluster similarly. Unfortunately, COMPAS recidivism tool has given different outputs for the individuals in the same cluster. In other words, COMPAS recidivism tool fails to satisfy individual fairness.

\subsubsection{Evaluating Group Fairness Notions \label{Sec: simulate gf}}
\begin{figure*}[t]
\centering
\includegraphics[width=.33\textwidth, trim={0 0 1cm 0},clip]{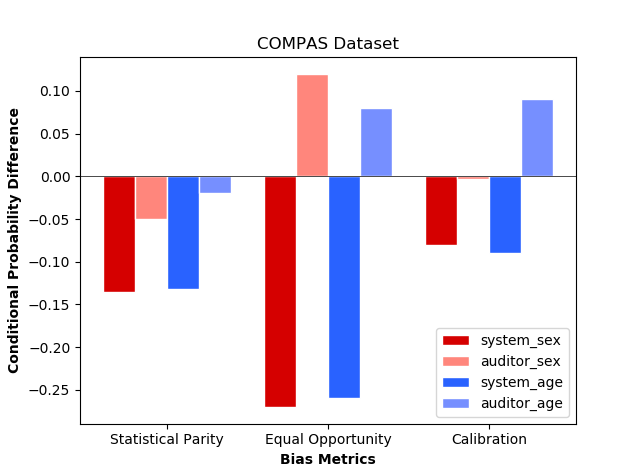}
\includegraphics[width=.33\textwidth, trim={0 0 1cm 0},clip]{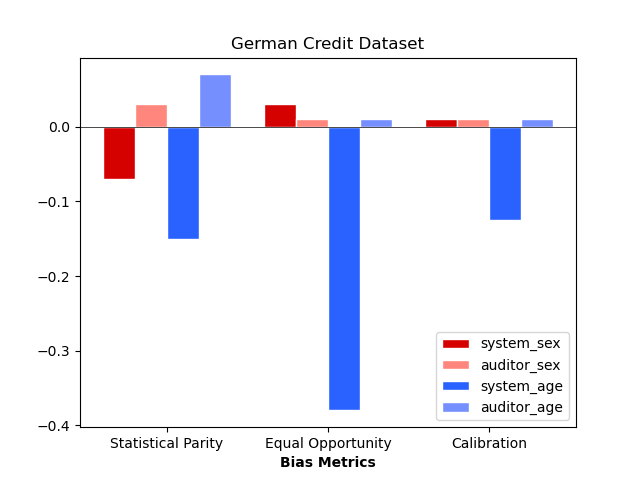}
\includegraphics[width=.33\textwidth, trim={0 0 1cm 0},clip]{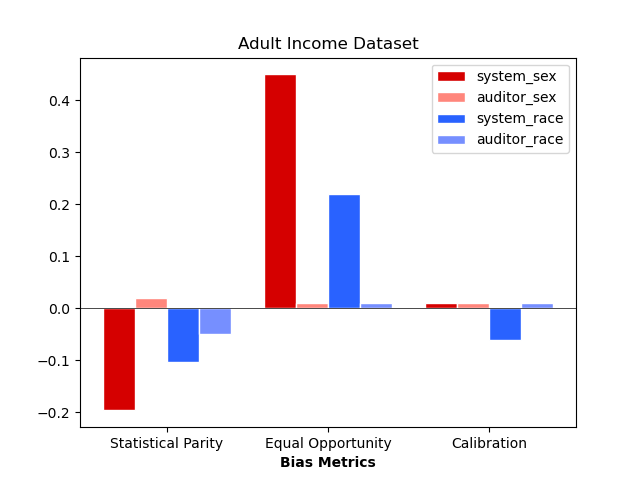}
\vspace{-2ex}
\caption{Real-world dataset vs. Auditor conditional probability differences across different group fairness notions with respect to sensitive attributes.}
\label{Fig: gf system-auditor comparison}
\vspace{-2ex}
\end{figure*}
We construct arbitrary intrinsic fair rules of the auditor for all the datasets discussed above. For COMPAS dataset, we consider \emph{number of prior offences} and \emph{degree of the offence} (felony or misdemeanor) to construct the fair relation. Note that, the binary feature \emph{two year recidivism} (most likely or least likely) is viewed as the output feature.
\begin{equation*}
f_1(x) =
\begin{cases}
1, & \text{if $x$.priors-count $\in [1,3]$ and $x$.charge-degree = F }
\\
& \qquad \qquad \qquad \qquad \text{OR}
\\
& \text{if $x$.priors-count $> 3$ and $x$.charge-degree = M}
\\[1ex] 
0, & \text{otherwise.} 
\end{cases}
\end{equation*}
Similarly, for German credit dataset, the features \emph{savings, credit history} and \emph{employment} are considered while designing the auditor's relation.
\begin{equation*}
f_2(x) = 
\begin{cases}
1, & \text{if $x$.savings $> 500$, $x$.credit-history = Paid, and $x$.employment $> 2$ years}
\\[1ex] 
2, & \text{otherwise.} 
\end{cases}    
\end{equation*}
Since the task of the Adult income dataset is to predict whether yearly income of an individual is $>$ 50K or $<=$ 50K, we consider the feature \emph{education} in the auditor's fair relation as shown below.
\begin{equation*}
f_3(x) = 
\begin{cases}
1, & \text{if $x$.education $\in$ [Bachelors, Masters, School Professor, Doctorate]}
\\[1ex] 
0, & \text{otherwise.} 
\end{cases}
\end{equation*}

Having defined auditor's evaluation functions, we now demonstrate whether the auditor is fair with respect to three group fairness notions (statistical parity, equalized odds, and calibration) as defined in Section \ref{sec: fairness notions} across every dataset. The conditional probability differences are computed between the unprivileged group to the privileged group with respect to the favorable outcome. The ideal value of this difference is 0 for all the three measures. In other words, if the probability difference is $>0$, the underprivileged group is benefited. Whereas, if the probability difference is $<0$, the privileged group is benefited. Figure \ref{Fig: gf system-auditor comparison} shows the probability differences of both the auditor and the respective real-world classifiers across different notions. We can observe that the simulated auditor assessments seem to perform well when compared to the real-world datasets. In case of COMPAS dataset, the probability differences of auditor assessments are close to zero compared to COMPAS itself. Moreover, the auditor seems to perform well in terms of statistical parity rather than equal opportunity or calibration. In other words, the auditor prefers statistical parity over calibration to evaluate COMPAS recidivism tool. Interestingly, the auditor's biases with respect to all the group fairness notions is almost equal to zero for German credit and Adult income datasets. For both the datasets, the auditor appears to prefer equal opportunity or calibration over statistical parity. Though we simulated auditor assessments using arbitrary functions, we strongly believe that \emph{expert} human auditors are capable of evaluating and performing well in terms of various fairness attributes. 

\subsection{Validation using Elicited Feedback from Real Human Auditors}
In this section, we evaluate the performance of various learning algorithms (e.g. logistic regression, decision trees and support vector machines) to estimate $f$ using real human feedback elicited by Dressel and Farid in \cite{dressel2018accuracy}. This data acquisition experiment consists of a short description of the defendant (gender, age, race, and previous criminal history) is provided to the human auditor. A total of 1000 defendant descriptions are used that are drawn randomly from the original ProPublica's COMPAS dataset. Furthermore, these descriptions were divided into 20 subsets of 50 each. The experiment consisted of 400 different crowd workers and each on of them was randomly assigned to see one of these 20 subsets. The participants predicted whether this particular individual would recidivate within 2 years of their most recent crime. The original data consists whether a crowd worker predicted correctly or not compared to the original classification. We preprocessed this dataset and obtained the true prediction by the crowd workers. As a preliminary analysis, we determine the number of crowd workers who satisfy different fairness notions. Figures \ref{Fig: crowd audit gf - race} and \ref{Fig: crowd audit gf - gender} show the percentage of crowd workers who satisfy three different group fairness notions as the probabilistic bias $\delta$ increases. In other words, as we relax the group fairness notions, the number of auditors satisfying the notions seem to increase both in terms of race and gender. Surprisingly, majority of the auditors satisfy statistical parity over calibration and equal opportunity showing clear preference among fairness notions. We observed similar results using simulated assessments as well in Section \ref{Sec: simulate gf}. However, there are only a few auditors who comply with group fairness notions when $\delta = 0$. In terms of individual fairness, we group the subsets provided to the auditors into clusters of similar defendant descriptions. The auditor is expected to give same classification for every defendant in that specific cluster. Figure \ref{Fig: crowd audit if} shows that no crowd worker satisfies individual fairness across every cluster present in a given subset. However, majority of them seem to comply with individual fairness for 80-90\% of the clusters.

\begin{figure}[t]
     \centering
     \begin{subfigure}{0.33\textwidth}
         \centering
         \includegraphics[width=\textwidth,trim={0 0 1cm 0},clip]{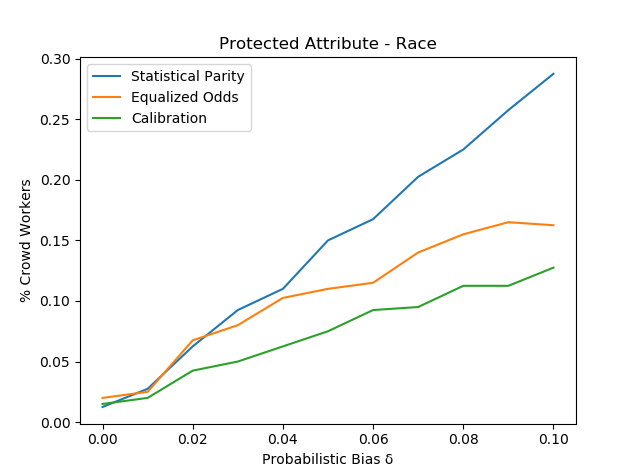}
         \caption{}
         \label{Fig: crowd audit gf - race}
     \end{subfigure}
     \begin{subfigure}{0.33\textwidth}
         \centering
         \includegraphics[width=\textwidth,trim={0 0 1cm 0},clip]{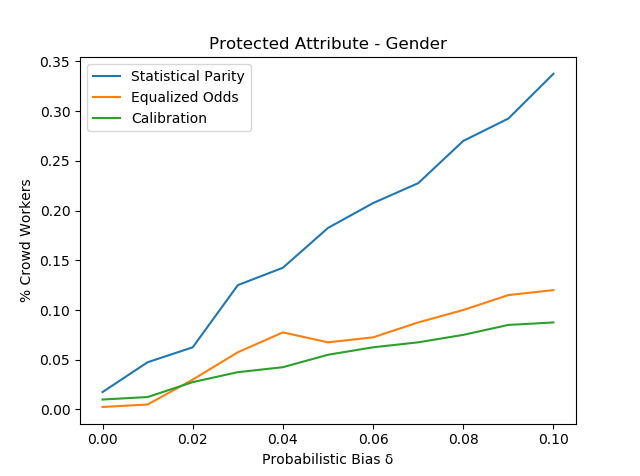}
         \caption{}
         \label{Fig: crowd audit gf - gender}
     \end{subfigure}
     \begin{subfigure}{0.33\textwidth}
         \centering
         \includegraphics[width=\textwidth, trim={0 0 1cm 0},clip]{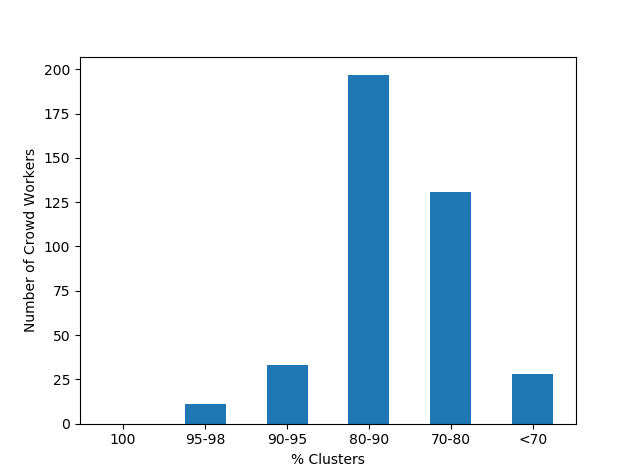}
         \caption{}
         \label{Fig: crowd audit if}
     \end{subfigure}
        \vspace{-2ex}
        \caption{\textbf{(a) and (b):} Percentage of crowd workers who satisfy respective group fairness notion as the probabilistic bias $\delta$ increases from 0 to 0.1 across the sensitive attributes race and gender. \textbf{(c):} Majority of the crowd workers satisfy individual fairness for 80\% of clusters present in the data.}
        \label{Fig: crowd audit gf vs if}
\end{figure}

Additionally, we learn and estimate crowd workers' responses using standard ML algorithms - logistic regression, decision tree, and SVM. Figure \ref{fig:predicting auditors} highlights the number of crowd workers whose responses are predicted using ML algorithms across different accuracy levels. It can be observed that logistic regression performs well by predicting the responses of 249 crowd workers with accuracy greater than 80\% compared to other models. Moreover, among the 249 crowd workers, majority of the them preferred calibration over statistical parity and equalized odds across every sensitive attribute. However, more than 60\% of the crowd workers did not prefer either of the group fairness notions while evaluating the given descriptions. Surprisingly, in case of the sensitive attribute \emph{race}, none of the crowd workers satisfy statistical parity and equalized odds. 

\begin{figure}
\begin{minipage}{0.45\linewidth}
% \hspace*{-1cm}
\includegraphics[width=\linewidth]{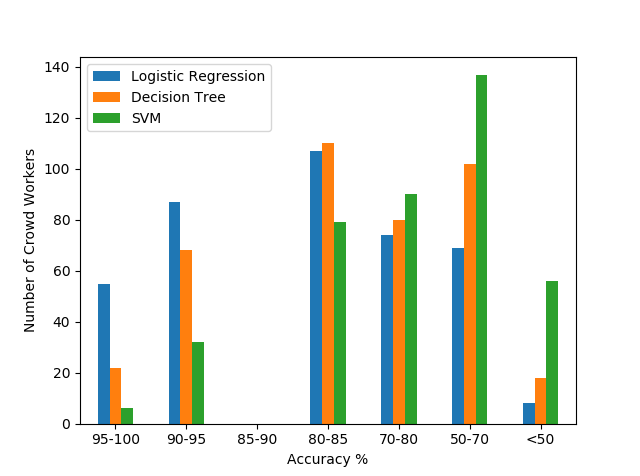}
\end{minipage}
\begin{minipage}{0.45\linewidth}
\begin{center}
\begin{adjustbox}{width=1\columnwidth,center}
\begin{tabular}{ |c|c|c|c|c| } 
\hline
& \multirow{3}{4em}{Statistical Parity} & \multirow{3}{4em}{Equalized Odds} & \multirow{3}{4.5em}{Calibration} & \multirow{3}{2.5em}{Other} \\
& & & & \\
& & & & \\
\hline
gender & 5.75\% & 2.25\% & 27.5\% & \textbf{64.5\%}\\ 
\hline
race & 0\% & 0\% & 9.5\% & \textbf{90.5\%} \\ 
\hline
\end{tabular}
\end{adjustbox} 
\end{center}
\end{minipage}
\caption{\textbf{Left:} The accuracies of different ML models on predicting the feedback of 400 crowd workers. We can see that logistic regression performed well compared to other models. \textbf{Right: } Among the 249 crowd workers whose responses are accurately predicted by logistic regression, majority of them satisfy calibration.}
\label{fig:predicting auditors}
\end{figure}

% \subsubsection{Aggregating Crowd Workers' Feedback}
% We use the following techniques to aggregate crowd workers' responses.
% \begin{enumerate}
%     \item Majority Decision \cite{von2008human}:
    
%     \item Honeypot \cite{lee2010social}:
% \end{enumerate}

\section{Conclusion and Future Work}
We developed a novel latent assessment model to characterize human auditor feedback and demonstrated its relationship with traditional fairness notions both theoretically and on real datasets. We obtained PAC learning guarantees on learning auditor's intrinsic fairness assessments, and demonstrated the learning performance of three learning algorithms on a real human feedback dataset. Consequently, this paper enabled us to accomplish two important challenges in the design of a crowd-auditing platform: (i) we can learn/mimic auditor's intrinsic evaluations using little elicited feedback and automate the evaluation on the remaining possibilities especially in high-dimensional learning algorithms, and (ii) we can also evaluate biases in auditor feedback with respect to various traditional fairness notions. 

In future, we will address all the other challenges in the design of crowd-auditing platforms. Specifically, we will use the relationship between LAM and traditional fairness notions to identify reliable auditors for feedback elicitation in our crowd-auditing platform design. Since feedback elicitation is an expensive process, we will improve our LAM model to account for feedback for data bundles, as opposed to our current feedback model for singleton data tuples. Furthermore, we will also investigate appropriate fusion rules to aggregate feedback collected from multiple auditor with heterogeneous opinions.

\clearpage
\bibliography{sample-bibliography}
\bibliographystyle{acm}

\end{document}